\documentclass[conference]{IEEEtran}
\usepackage{amsmath,amsfonts,amsthm}

\newcommand{\matr}[1]{\bm{#1}}
\newcommand{\set}[1]{\mathcal{#1}}
\usepackage{bm}
\usepackage{cite}
\newcommand{\vect}[1]{\bm{#1}}
\theoremstyle{definition}
\newtheorem{theorem}{Theorem}
\newtheorem{defn}{Definition}[section]
\newtheorem{remark}{Remark}

\newtheorem{lemma}[]{Lemma}
\usepackage{bbm}
\usepackage{algorithmic}
\usepackage{calc}
\usepackage{array}
\usepackage{tabularx}
\usepackage[font=normalsize,labelfont=sf,textfont=sf]{subcaption}
\usepackage{textcomp}
\usepackage{stfloats}
\usepackage{url}
\usepackage{verbatim}
\usepackage{graphicx}
\usepackage{xcolor}
\def\BibTeX{{\rm B\kern-.05em{\sc i\kern-.025em b}\kern-.08em
    T\kern-.1667em\lower.7ex\hbox{E}\kern-.125emX}}
\usepackage{balance}
\usepackage{hyperref}

\begin{document}
\title{On the Impact of Sample Size in Reconstructing Graph Signals}
\author{\IEEEauthorblockN{Baskaran Sripathmanathan, Xiaowen Dong, Michael Bronstein}
    \IEEEauthorblockA{University of Oxford}
}


\maketitle

\begin{abstract}
Reconstructing a signal on a graph from observations on a subset of the vertices is a fundamental problem in the field of graph signal processing. It is often assumed that adding additional observations to an observation set will reduce the expected reconstruction error. We show that under the setting of noisy observation and least-squares reconstruction this is not always the case, characterising the behaviour both theoretically and experimentally. 
\end{abstract}


\section{Introduction}

\IEEEPARstart{G}{raph} signal processing (GSP) has gained popularity owing to its ability to process and analyze signals on graphs, such as political preferences \cite{renoust2017estimating}, brain fMRIs \cite{itani2021graph} and urban air pollution \cite{jain2014big}. GSP generalises the highly successful tools of classical signal processing from regular domains such as grids to graphs.
Similar to the classical case, the computational costs of processing and storing large volumes of graph signals can be prohibitive, and complete data may not be available owing to impractically high observation costs. Graph sampling provides a solution to these problems by efficiently extrapolating the full data across the graph from observations on a set of vertices or summaries of the data \cite{tanaka2020sampling}.

Sampling in the graph setting poses more challenges than classical sampling because of the irregularity of the graph domain. One such challenge is that periodic sampling, widely used in traditional signal processing, is not applicable. Instead, sample selection must adapt to the graph's topology. Optimal sample selection on graphs is in general NP-hard \cite{chamon2017greedy, nikolov2022proportional}. Many works focus on providing efficient heuristics for selecting good sample sets under different optimality criteria.
These studies also provide bounds to help practitioners manage the trade-off between observation cost and reconstruction loss while determining sample size \cite{wang2018optimal,wang2019low, mfn, jayawant2021doptimal, tremblay2017determinantal, bai2020fast}. One limitation of these bounds is the scope of their settings: some bounds are set in the noiseless setting \cite{shomorony2014sampling}, while most recent sample-set selection literature is set in the noisy observation setting. In the noisy observation setting, sample-size bounds can require optimal Bayesian reconstruction \cite{chamon2017greedy}, as opposed to being generic over the various reconstruction methods presented in and benchmarked against in the sample set selection literature, e.g., least squares (LS) \cite{jayawant2021doptimal, tremblay2017determinantal}, variants of LS  \cite{wang2018optimal, wang2019low} or graph-Laplacian regularised (GLR) reconstruction \cite{bai2020fast}. Furthermore, driven by these bounds, many papers in the sample set literature only present experiments with sample sizes exceeding the bandwidth.

This paper presents two primary contributions. First, we demonstrate that the commonly held expectation that increasing sample size results in lower MSE (presented, for example, below equation (13) in \cite{chamon2017greedy}) does not hold under LS in many of the settings studied in the literature for signals with noise. Second, we show that it is possible to simultaneously reduce observation cost and reconstruction error compared to sampling and reconstruction schemes presented in the literature. We support our findings with theoretical evidence and experiments conducted under LS. 

\section{Related Work}
\label{app:related_work}
Graph signal processing extends the fundamental problem of sampling and reconstruction from signals in the Euclidean domain to graph-structured data. It does so by generalising the graph shift operator \cite{ortega2018graph} - the most common choices being the adjacency matrix of the graph, the graph Laplacian, or a normalised variant of those - and using it to define a signal model. While some work uses the adjacency matrix as the shift operator \cite{EOptimalChen}, and the theorems in \cite{chamon2017greedy} apply to all of these operators, most of the literature uses a normalised variant of the graph Laplacian. See \cite{ortega2018graph} for a more complete consideration of the trade-offs involved in this choice.

The graph sampling literature is further divided by considerations around the signal model, the reconstruction method and the optimality objective, which we describe below.

\subsection{Bandlimited Signals}
The most common signal model used in the literature is the bandlimited signal model. For a graph $\mathcal{G}$ with a shift operator $\matr{L}$ with eigenvalues $\lambda_1 \leq .. \leq \lambda_N$, the space spanned by the first $k$ eigenvectors of $\matr{L}$ is called a Paley-Wiener space $PW_{\omega}(\mathcal{G})$ (for any $\omega \in  [\lambda_{k}, \lambda_{k+1})$) and its elements are called $k$-bandlimited signals. Pesenson \cite{pesenson2008sampling} introduced the concept of a \emph{uniqueness set} which is a vertex set capable of perfectly reconstructing any signal in $PW_{\omega}(\mathcal{G})$, and notes that it must include at least $k$ vertices. This provides a unique optimality criterion for sample sets, for which multiple sampling schemes have been devised \cite{anis2016efficient,puy2018random}.

\subsection{Non-bandlimited signals}
It is rare for observed signals to be perfectly bandlimited. While some work has focused on extending the class of underlying signals to `approximately bandlimited signals' \cite{chen2016signal, lin2019active}, it is mostly assumed that there is a clean underlying bandlimited signal and our observations are corrupted by additive noise. The extra error introduced by this noise is handled in two ways: noise-aware sampling criteria (and corresponding sampling schemes), and robust reconstruction algorithms.

While there is a unique optimality criterion in the noiseless case, there are multiple in the additive noise case:

\begin{itemize}
    \item \emph{MMSE criterion}: Minimise the average mean squared error (MSE), called \emph{A-Optimality} under LS \cite{wang2018optimal,wang2019low, mfn}
    \item \emph{Confidence Ellipsoid criterion}: Minimise the confidence ellipsoid around the eigenbasis co-efficients, which is called \emph{D-optimality} under LS \cite{jayawant2021doptimal, tremblay2017determinantal}
    \item \emph{WMSE criterion}: Minimise the worst-case MSE, which is called \emph{E-optimality} under LS \cite{bai2020fast, EOptimalChen}
\end{itemize}

These criteria under LS reconstruction, and equivalences to other optimality criteria, are further studied in the literature on Optimal Design of Experiments (see \cite{wang2018optimal} for more detail).

\subsection{Reconstruction Methods}
Reconstruction methods are also known as interpolation operators \cite{chamon2017greedy}. The most common methods of reconstructing noisy signals are LS reconstruction \cite{jayawant2021doptimal, tremblay2017determinantal}, variants of LS reconstruction \cite{wang2018optimal, wang2019low} and GLR reconstruction \cite{bai2020fast}. The variants of LS reconstruction and GLR reconstruction are more robust than ordinary LS reconstruction. We provide more detail on these schemes in Section \ref{sec:rec-methods}.

\section{Preliminaries}
\subsection{Graph Signals and bandlimitedness}
{\color{black}We define a graph $\mathcal{G}$ to consist of a set of $N$ vertices and a set of edges with associated edge weights. We assume the graph is connected and undirected.}
We consider a \emph{bandlimited} signal $\vect{x}$ on $\mathcal{G}$, generalising the classical signal processing definition of bandlimitedness. We do so by considering a symmetric, positive semi-definite shift operator $\matr{L}$ on $\mathcal{G}$; commonly used examples of $\matr{L}$ include the combinatorial Laplacian and its normalised variants. We take its eigendecomposition
\begin{equation}
    \matr{L} = \matr{U}\matr{\Lambda} \matr{U}^{T} \nonumber
\end{equation}
writing $\matr{\Lambda} = diag(\lambda_1, \dots, \lambda_N)$ where $0 = \lambda_{1} \leq \dots \leq \lambda_{N}$ are the eigenvalues of $\matr{L}$, also known as the \emph{graph frequencies}\cite{ortega2018graph}. $\matr{U}$ forms an orthonormal basis of $\mathbb{R}^{N}$; let $\matr{U}_{k}$ be the $k$ columns of $\matr{U}$ corresponding to the $k$ lowest graph frequencies. We say that $\vect{x}$ is \emph{$k$-bandlimited} if $\vect{x} \in span(\matr{U}_{k})$.

\subsection{Sampling}
Inherent to the definition of bandlimitedness is that $\vect{x}$ comes from a low-dimensional subspace. This implies that we do not need to observe $\vect{x}$ on all $N$ vertices. Indeed, there is some subset of vertices such that if we observe \emph{any} $k$-bandlimited $\vect{x}$ on that subset we can reconstruct $\vect{x}$ fully and without error. Such a subset is called a \emph{uniqueness set} \cite{pesenson2008sampling}.

Given a vertex sample set $\set{S}$, let $\matr{M}_{\set{S}} \in \mathbb{R}^{|\set{S}| \times N}$ be the corresponding sampling matrix where
\begin{equation}
    \left(\matr{M}_{\set{S}}\right)_{ij} =
        \begin{cases}
      1 & \text{if}\ \set{S}_{i} = j \\
      0, & \text{otherwise}
    \end{cases}
\end{equation}
then $\set{S}$ is a uniqueness set for a bandwidth $k$ if and only if $\text{rank}(\matr{M}_\set{S}\matr{U}_{k}) = k$ \cite{anis2016efficient}.

In practice, the signal we are given is often not perfectly bandlimited. We model this as observation noise; we observe a corrupted signal $\vect{y} = \vect{x} + \sigma \cdot \vect{\epsilon}$ where
\begin{itemize}
    \item $\vect{x} \sim \mathcal{N}(\vect{0}, \matr{U}_{k}\matr{U}_{k}^{T})$ is a $k$-bandlimited Gaussian signal,
    \item $\vect{\epsilon} \sim \mathcal{N}(\vect{0},\matr{I}_{N})$ is i.i.d. Gaussian noise on each vertex,
    \item $\sigma > 0$ is some scaling of the noise
\end{itemize}
so the corrupted signal $\vect{y}$ has high-frequency components. Let the Signal-to-Noise ratio SNR = $\frac{\text{tr}(\text{cov}(\vect{x}))}{\text{tr}(\text{cov}(
\sigma \cdot \vect{\epsilon}))}$ be the ratio of the variance of the signal to the variance of the noise. Then as a ratio of variances, SNR is positive\footnote{It is common in the literature to express the SNR in decibels, which may be negative, while its ratio form remains positive. We will only use the ratio form, so for example $-20dB$ would be written as $10^{-20/10} = 10^{-2} > 0$.} and $\sigma^{2} = \frac{k}{N \cdot \text{SNR}}$.

There are multiple optimality criteria in the literature for the noisy setting; under LS they have the following forms:
\begin{align}
    \text{A-Optimality:} &\text{ minimise } \text{tr}(((\matr{M}_{\set{S}}\matr{U}_{k})(\matr{M}_{\set{S}}\matr{U}_{k})^{T})^{-1}) \label{eq:A-optimality}\\
    \text{D-Optimality:} &\text{ maximise } \text{det}((\matr{M}_{\set{S}}\matr{U}_{k})(\matr{M}_{\set{S}}\matr{U}_{k})^{T})\label{eq:D-optimality}\\
    \text{E-Optimality:} &\text{ maximise }
        \lambda_{min}((\matr{M}_{\set{S}}\matr{U}_{k})(\matr{M}_{\set{S}}\matr{U}_{k})^{T}) \label{eq:E-optimality}
\end{align}
\noindent where $\lambda_{min}(\matr{A})$ is the minimum eigenvalue of $\matr{A}$.

In this paper, we use average MSE under our model as our loss, which corresponds to A-optimality for LS. 

\subsection{Reconstruction Methods}
\label{sec:rec-methods}
There exist two common reconstruction methods in the literature: LS reconstruction (a.k.a. the standard decoder \cite{puy2018random}) and GLR reconstruction (as described in \cite{puy2018random,bai2020fast}). We summarise the differences in Table \ref{tab:tbl1}. Our analysis of LS also applies to the commonly used iterative method, Projection onto Convex Sets \cite{narang2013localized}, 
as POCS converges to LS.

\begin{table}[h]
\caption{Reconstruction Methods}
\begin{center}
    \begin{tabularx}{(\textwidth - 12pt)/2}{| X | c |X|X|X|}
    \hline
     & Objective & Param & Bias & Needs $\matr{U}_{k}$ \\
    \hline
        LS  & $\displaystyle \min_{\vect{x} \in span(\matr{U}_{k})} || \matr{M}\vect{x} - \vect{y} ||_{2}$ & band-width $k$ & no & yes \\
        \hline
        GLR & $\displaystyle \min_{\vect{x} \in \mathbb{R}^{N}} \left( || \matr{M}\vect{x} - \vect{y} ||_{2} + \mu \vect{x}^{T}\matr{L}\vect{x} \right)$ & $\mu$ & yes & no \\
        \hline
    \end{tabularx}
\end{center}
\label{tab:tbl1}
\end{table}

It is well known that, for linear models with noise, LS reconstruction is the minimum-variance unbiased estimator of $\vect{x}$ \cite{gauss1823theoria}. This justifies us focusing our analysis of unbiased linear reconstruction methods on LS, at least theoretically. In practice, computing $\matr{U}_{k}$ is slow, so GLR reconstruction is used for large graphs instead \cite{puy2018random, bai2020fast}.

We define a \emph{reconstruction method} to take observations on a vertex sample set $\set{S}$ and reconstruct the signal across all vertices. We say that a reconstruction method is \emph{linear} if it is linear in its observations. 
For a fixed vertex sample set $\set{S}$ we can represent a linear reconstruction method by a matrix $\matr{R}_{\set{S}} \in \mathbb{R}^{N \times |\set{S}|}$.

\begin{remark}
    LS and GLR reconstruction are both linear:
    \begin{align}
    \textrm{LS:\quad} \matr{R}_{\set{S}} &= \matr{U}_{k}(\matr{M}_{\set{S}}\matr{U}_{k})^{\dagger} \\
    \textrm{GLR:\quad} \matr{R}_{\set{S}} &= (\matr{M}_{\set{S}}^{T}\matr{M}_{\set{S}} + \mu \matr{L})^{-1}\matr{M}_{\set{S}}^T
    \end{align}
\end{remark}
\noindent where for a matrix $\matr{A}$, $\matr{A}^{\dagger}$ is its Moore-Penrose pseudoinverse.

{\color{black}\section{Problem Setting}
For our theoretical results and experiments, we assume:
\begin{itemize}
    \item A clean underlying $k$-bandlimited signal $\vect{x}$.
    \item The bandwidth $k$ is known.
    \item Observations of the signal are corrupted by flat-spectrum noise, meaning we observe a non-bandlimited signal.
    \item We focus on LS reconstruction.
\end{itemize}

Note that when the sample size is below the bandwidth, there are often multiple possible reconstructions. For example, when trying to minimise the LS criterion \[\min_{z \in span(\matr{U}_{k})}\left|\left| \matr{M}_{\set{S}}\vect{z} - \vect{y}\right|\right|\]
the following is a solution for any $\vect{\delta} \in \mathbb{R}^{N}$:
\[ \matr{U}_{k}\left( \left( \matr{M}_{\set{S}}\matr{U}_{k}\right)^{\dagger}\vect{z} + (\matr{I} -  \left( \matr{M}_{\set{S}}\matr{U}_{k}\right)^{\dagger} \left( \matr{M}_{\set{S}}\matr{U}_{k}\right) \vect{\delta}\right).\]

As we are mainly concerned in studying how sample size affects reconstruction error rather than recommending a specific reconstruction algorithm, for simplicity we pick the minimal 2-norm solution with $\vect{\delta} = 0$. This uniquely defines LS reconstruction even when $|\set{S}| < k$ \cite[Sect. 5.5.1]{golub13}.}

\section{Main Results}
\label{main_results_sec}
Consider reconstructing a signal with LS reconstruction. We observe the corrupted signal $\vect{y}$ at $\set{S}$
 and reconstruct $\vect{x}$ ($\hat{\vect{x}} = \matr{R}_{\set{S}}\matr{M}_{\set{S}}\vect{y}$). We decompose the expected MSE from observing $\vect{y}$ at $\set{S}$:
\begin{align}
    \mathbb{E}[\text{MSE}_{\set{S}}]
    &= \mathbb{E}\left[ || \vect{x} - \matr{R}_{\set{S}} \matr{M}_{\set{S}} \vect{y} ||^{2}_{2} \right] \nonumber \\
    &= \text{tr}(\text{cov}(\vect{x} - \matr{R}_{\set{S}} \matr{M}_{\set{S}} \vect{y})) \nonumber \\
        &=  \text{tr}(\text{cov}(\vect{x} - \matr{R}_{\set{S}} \matr{M}_{\set{S}} (\vect{x} + \sigma \cdot \vect{\epsilon})))  \label{eq:EMSE_midpt}
\end{align}
    As the underlying signal $\vect{x}$ and the noise $\vect{\epsilon}$ are independent, and as our sampling and reconstruction operators are linear:
    \begin{equation}
    \begin{split}
        \text{cov}(\vect{x} - \matr{R}_{\set{S}} \matr{M}_{\set{S}} (\vect{x} + \sigma \cdot \vect{\epsilon}))  =  \text{cov}((\vect{x} - \matr{R}_{\set{S}} \matr{M}_{\set{S}} \vect{x})) \\ + \text{cov}( \sigma \cdot (\matr{R}_{\set{S}} \matr{M}_{\set{S}} \vect{\epsilon})).
        \label{eq:split_noise_and_signal}
    \end{split}
    \end{equation}
    Let $\matr{E} = \matr{I} - \matr{R}_{\set{S}} \matr{M}_{\set{S}}$ and combine (\ref{eq:EMSE_midpt}) and (\ref{eq:split_noise_and_signal}): 
\begin{alignat}{3}
    {\color{black}\mathbb{E}[\text{MSE}_{\set{S}}]}
    &= & {\color{black}\text{tr}( \text{cov}( \matr{E} \vect{x} ) )}  
        &+ {\color{black}\sigma^{2} \cdot  \text{tr}( \text{cov}( \matr{R}_{\set{S}} \matr{M}_{\set{S}} \vect{\epsilon} ) )}  \nonumber \\
    &= &\text{tr}( (\matr{E}\matr{U}_{k})(\matr{E}\matr{U}_{k})^{T})
        &+ \sigma^{2} \cdot \text{tr}(\matr{R}_{\set{S}} \matr{M}_{\set{S}} \matr{M}_{\set{S}}^{T} \matr{R}_{\set{S}}^{T}) \nonumber \\
    &= &||\matr{E}\matr{U}_{k}||^{2}_{F}
        &+ \sigma^{2} \cdot \text{tr}(\matr{R}_{\set{S}} \matr{R}_{\set{S}}^{T}) \nonumber \\
    &= &|| \matr{U}_{k} - \matr{R}_{\set{S}}\matr{M}_{\set{S}}\matr{U}_{k} ||^{2}_{F} 
        &+ \sigma^2 \cdot || \matr{R}_{\set{S}} ||^{2}_{F}.
\end{alignat}

\noindent We define
    \begin{align}
        \xi_{1}(\set{S}) &= || \matr{U}_{k} - \matr{R}_{\set{S}}\matr{M}_{\set{S}}\matr{U}_{k} ||^{2}_{F} \\
        \xi_{2}(\set{S}) &= || \matr{R}_{\set{S}} ||^{2}_{F}
    \end{align}
so that 
\begin{equation}
    \mathbb{E}[\text{MSE}_{\set{S}}] = \xi_{1}(\set{S}) + \sigma^{2} \cdot \xi_{2}(\set{S}). \label{eq:xi_decomp}
\end{equation}
\begin{remark}
    
Setting $\sigma=0$, we see that $\xi_{1}(\set{S})$ can be interpreted as the reconstruction error in the absence of observation noise.
\end{remark}

We use this decomposition to analyse changing the vertex sample set $\set{S}$. We consider removing a vertex $v$ from $\set{S}$ to make $\set{S} \backslash \{v\}$. 

\begin{defn}
Removing $v$ \emph{improves} $\set{S}$ if $$\mathbb{E}[\text{MSE}_{\set{S}}] > \mathbb{E}[\text{MSE}_{\set{S} \backslash \{v\}}].$$
\end{defn}

For $i \in \{1,2\}$, let
\begin{equation}
    \Delta_i(\set{S}, v) = \xi_{i}(\set{S}) - \xi_{i}(\set{S} \backslash \{ v\}). 
\end{equation}
Then by (\ref{eq:xi_decomp}), the change in MSE from removing $v$ is 
\begin{equation}
    \mathbb{E}[\text{MSE}_{\set{S}}] - \mathbb{E}[\text{MSE}_{\set{S} \backslash \{v\}}]
    = \Delta_{1}(\set{S},v) + \sigma^{2} \cdot \Delta_{2}(\set{S},v). \label{eq:EMSE_decomp_into_delta}
\end{equation}
If $\Delta_{1}(\set{S},v) + \sigma^{2} \cdot \Delta_{2}(\set{S},v) > 0$, removing $v$ improves $\set{S}$.

We note that $\Delta_{1}$ is the change in MSE when there is no noise ($\sigma = 0$), so can be interpreted as learning more about $\vect{x}$. It is always non-positive under LS reconstruction (see Appendix \ref{app:delta_1_non_positive}). On the other hand, $\Delta_{2}$ is a noise-sensitivity term --- its effect scales with $
\sigma^{2}$ -- and in many cases is positive.
Under LS reconstruction, one can show that $\Delta_1$ and $\Delta_2$ are always of different signs (see Appendix \ref{app:Delta_LS_opposite_signs}).

If the effect of increasing noise sensitivity exceeds the effect of learning more about the underlying signal, then we can decrease average MSE by \emph{removing} a vertex from the observation set. This leads to our main result under LS reconstruction, which is summarised in the following theorem:

\begin{theorem}
\label{main_ls}
    Let 
    \begin{equation}
        \tau(\set{S},v) = \frac{k}{N} \cdot \Delta_{2}(\set{S},v)
    \end{equation}
    then removing $v$ improves $\set{S}$ if and only if
 \begin{equation}
 \label{eq:LS_corol_iff}
     \textrm{SNR} < \tau(\set{S},v).
 \end{equation}
 \end{theorem}
\begin{proof}
    See Appendix \ref{proof_appendix_LS}.
\end{proof}

This result says that if SNR is too low (below a threshold $\tau$ that depends on the bandwidth and the chosen samples), then we can remove a sample from our observation set to improve the average reconstruction error.

\begin{remark}
    If $\Delta_{2}$ is non-positive, we have $\tau(\set{S},v) \leq 0 < $ SNR. In this case (\ref{eq:LS_corol_iff}) cannot hold and removing $v$ will not improve $\set{S}$ for any SNR.
\end{remark}

Theorem \ref{main_ls} leaves room for a clever way to pick vertices such that the conditions on SNR in (\ref{eq:LS_corol_iff}) would never be met, hence removing a vertex would never improve the sample set. We show that no such way exists.

\begin{theorem}
    \label{thm:main_existence_LS}
    Consider a fixed vertex ordering $v_{1},\dots, v_{N}$ and let $\set{S}_{i}$ be the set of the first $i$ vertices. Then there are {\color{black}exactly} $k$ indices $1 \leq I_{1}, \dots, I_{k} \leq N$ such that
    \begin{equation}
        \forall 1\leq j \leq k: \tau(\set{S}_{I_{j}}, v_{I_{j}}) > 0,
    \end{equation}
    so removing $v_{I_{j}}$ improves $\set{S}_{I_{j}}$ at some SNR.
\end{theorem}
\begin{proof}
    See Appendix \ref{app:LS_satisfied}.
\end{proof}
    Theorem \ref{thm:main_existence_LS} suggests that any sampling scheme, interpreted as a sequential way of picking additional samples, must encounter exactly $k$ instances where the additional vertex $v$ picked on top of the current sample set $\set{S}$ has $\tau(\set{S}\cup v, v) > 0$, meaning at a high enough noise level it increases MSE on average. Schemes in the literature which are optimal in the noiseless case, such as A-, D- and E-optimal sampling schemes, see this happen for the first $k$ vertices they pick.
\begin{theorem}
\label{thm:noiseless_optimality_means_noise_sensitivity}
    Suppose we have a greedy scheme which is optimal in the noiseless case: given the bandwidth $k$, the first $k$ vertices it samples allow for perfect reconstruction of any clean $k$-bandlimited signal. Use this scheme to select a vertex sample set $\set{S}_{m}$ with $|\set{S}_{m}| = m \leq k$. Then
    \begin{equation}
    \label{eq:greedy_sampling_first_k}
        {\color{black} \forall m \leq k: \enskip}  \forall v \in \set{S}_{m}: \enskip \tau(\set{S}_{m},v) > 0,
    \end{equation}
    that is, for \emph{any} vertex in $\set{S}$, there exists some SNR such that removing that vertex would improve $\set{S}$. {\color{black} Removal of any vertices which the scheme adds after this cannot improve the set: }
    \begin{equation}
        \label{eq:greedy_sampling_over_k}
        {\color{black} \forall m' > k: \enskip \forall v \in \set{S}_{m'} \backslash \set{S}_{k}: \enskip \tau(\set{S}_{m'},v) \leq 0.}
    \end{equation}
\end{theorem}
\begin{proof}
    See Appendix \ref{app:optimal_noiseless_schemes_immediately_satisfy}.
\end{proof}

\begin{remark}
\label{remark:ADE_are_noiseless_optimal}
    Theorem \ref{thm:noiseless_optimality_means_noise_sensitivity} applies to A, D and E-optimal sampling as they are optimal in the noiseless case (see Appendix \ref{app:proof_of_remark_ADE_are_noiseless_optimal}). 
    \end{remark}
\begin{remark}
{\color{black} 
    Equation (\ref{eq:greedy_sampling_over_k}) says that removing one vertex from a sample set of size $m' > k$ chosen by a noiseless-optimal sampling scheme does not reduce error on average. Of course, if one removes \emph{multiple} vertices to bring the sample size below $k$ then the expected sample error may decrease.
}
\end{remark}


\section{Experiments}
\label{experiments_sec}

\begin{figure*}%
    \label{LS_ER_MSE_fig}
    \centering
    \begin{subfigure}{0.6\columnwidth}
    \includegraphics[width=\columnwidth]{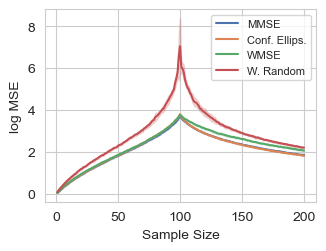}
    \caption{SNR = $10^{-1}$}
    \label{MSE_subfiga}
    \end{subfigure}\hfill
    \begin{subfigure}{0.6\columnwidth}
    \includegraphics[width=\columnwidth]{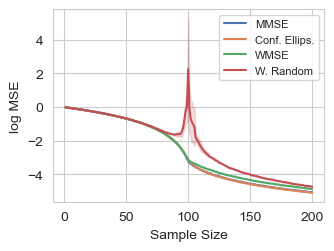}%
    \caption{SNR = $10^{2}$}%
    \label{MSE_subfigb}%
    \end{subfigure}\hfill%
    \begin{subfigure}{0.6\columnwidth}
    \includegraphics[width=\columnwidth]{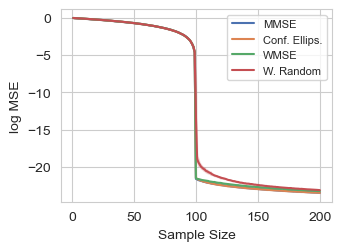}%
    \caption{SNR = $10^{10}$}%
    \label{MSE_subfigc}%
    \end{subfigure}%
    \caption{Average MSE for LS reconstruction on ER Graphs (\#vertices=1000, bandwidth = 100) with different SNRs}
\label{LS_ER_MSE_fig}
\end{figure*}

\begin{figure*}%
    \centering
    \begin{subfigure}{0.6\columnwidth}
    \includegraphics[width=\columnwidth]{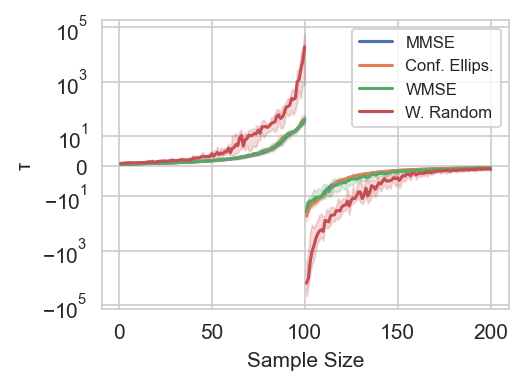}
    \caption{Erdos-Renyi}
    \label{snr_ER}
    \end{subfigure}
    \hfill
    \begin{subfigure}{0.6\columnwidth}
    \includegraphics[width=\columnwidth]{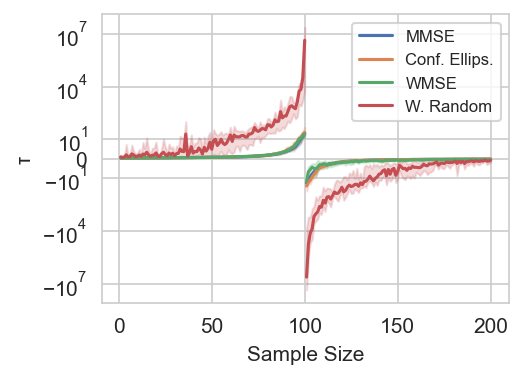}%
    \caption{Barabasi-Albert}%
    \label{snr_BA}%
    \end{subfigure}
    \hfill%
    \begin{subfigure}{0.6\columnwidth}
    \includegraphics[width=\columnwidth]{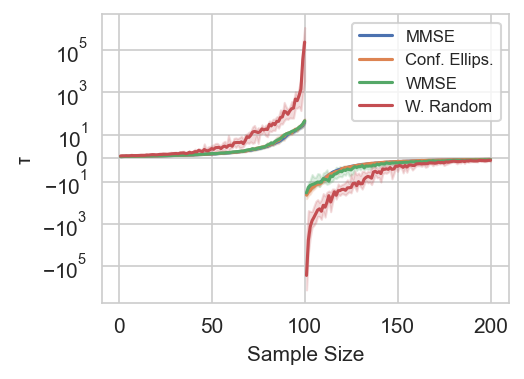}%
    \caption{SBM}%
    \label{snr_SBM}%
    \end{subfigure}%
    \caption{$\tau$ for different random graph models under LS reconstruction (\#vertices = 1000, bandwidth = 100)}
\label{LS_SNR_Threshold_plots}
\end{figure*}

\subsection{Experimental Setup}
We present two experiments to illustrate when removing vertices from the observation set can reduce MSE. For different types of graphs, we present plots of  {\color{black} $\mathbb{E}[\textrm{MSE}_{\set{S}_{i}}]$ (Fig. \ref{LS_ER_MSE_fig}) and $\tau(\set{S}_{i},v_{i})$ (Fig. \ref{LS_SNR_Threshold_plots}) as the sample size $i$ increases under different sampling schemes. Results are presented with 90\% confidence intervals.}
\subsubsection{Sample Set Selection}
The literature provides several approximations to make vertex sample set selection efficient. For example, approximating the projection matrix $\matr{U}_{k}\matr{U}_{k}^{T}$ \cite{wang2018optimal} (subsets of which are used to compute optimality criteria) with a polynomial in $\matr{L}$, and approximating optimality criteria for easier computation \cite{bai2020fast}.

For our experiments, we generate the vertex sample sets greedily using the exact analytical forms instead of approximations. We use the explicit forms of A/D/E optimality (see Eqns. (\ref{eq:A-optimality}), (\ref{eq:D-optimality}), (\ref{eq:E-optimality})) and directly compute $\matr{U}_{k}\matr{U}_{k}^{T}$ throughout. We compare A/D/E optimal schemes (MMSE/Conf. Ellips./WMSE) to the Weighted Random sampling scheme in \cite{puy2018random}.

\subsubsection{Graph Generation}
We consider two graph sizes -- small (100 vertices) and large (1000 vertices) -- for 10 instantiations of each of the following {\color{black}unweighted} random graph models:
\begin{itemize}
    \item Erdős–Rényi (ER) with edge probability $p=0.8$
    \item Barabási-Albert (BA) with a preferential attachment to 3 vertices at each step of its construction
    \item Stochastic Blockmodel (SBM) with intra- and inter-cluster edge probabilities of $0.7 \text{ and }0.1$ respectively
\end{itemize}

\subsubsection{Signal Generation}
To compute the MSE, we generate 200 signals as follows:
\begin{enumerate}
    \item Generate $\vect{x}_{raw} \sim \mathcal{N}(\vect{0}, \matr{U}_{k}\matr{U}_{k}^{T})$, $\vect{\epsilon}_{raw} \sim \mathcal{N}(\vect{0}, \matr{I}_{N})$
    \item Normalise: $\vect{x} = \frac{\vect{x}_{raw}}{||\vect{x}_{raw}||_{2}}$ and $\vect{\epsilon} = \frac{\vect{\epsilon}_{raw}}{ ||\vect{\epsilon}_{raw}||_{2}}$
    \item Return $\vect{y} = \vect{x} + \frac{\vect{\epsilon}}{\sqrt{\text{SNR}}}$
\end{enumerate}

\subsubsection{Parameters}
We set the bandwidth to $\lfloor N / 10 \rfloor$, as per \cite{bai2020fast}.  We pick various SNRs to demonstrate that the effect occurs below the threshold $\tau$ and disappears above it, which here means $10^{-1}, 10^{2}, 10^{10}$ (In $dB$: $-10dB$, $20dB$ and $100dB$).

\subsection{Experimental Results}
{\color{black}
Fig. \ref{MSE_subfiga} shows that for low SNRs, optimal sampling schemes (the Green, Orange and Blue lines) lead MSE to increase with each additional sample until the sample size reaches the bandwidth, illustrating Theorem \ref{thm:noiseless_optimality_means_noise_sensitivity}. 
On the other hand, for high SNRs (Fig. \ref{MSE_subfigc}), MSE decreases monotonically as sample size increases under all presented sampling schemes, illustrating Theorem \ref{main_ls}. Fig \ref{MSE_subfigb} shows an intermediate case: for SNRs between these extremes some schemes (Red line) lead to increasing MSE with increasing sample size, while other schemes (Blue, Green, Orange lines) do not.

Interestingly, Fig. \ref{MSE_subfiga} shows that at a very low SNR of $10^{-1}$, the optimal sample size under LS reconstruction is zero. One interpretation of this observation is that, under very high noise, if you throw away all of your samples and assume that your underlying signal is $\vect{0}$, you will on average have a lower MSE than if you reconstruct with LS from your observed samples. This follows from (\ref{eq:xi_decomp}) and the positivity of $\xi_{1}$ and $\xi_{2}$ -- if your error increases unboundedly with noise, at a sufficiently high noise level your MSE will be above the fixed MSE you would get by approximating your signal with $\vect{0}$.}

Fig. \ref{snr_ER} demonstrates experimentally {\color{black}what the SNR threshold $\tau$ looks like in practice. F}or the ER graphs ($N = 1000)$, for signals with $k=100$ and $\textrm{SNR} < 10$, under most `optimal' schemes (Blue, Green, Orange) sampling 90 vertices rather than 100 vertices reduces both observation cost and reconstruction error. Similar results can be seen for Barabasi-Albert graphs (Fig. \ref{snr_BA}) and different SNRs. This demonstrates that MSE increasing with sample size happens under conditions which might occur in practice, and is not simply a theoretical curiosity.

We present plots for the larger graph instances here; the smaller graphs (Fig. \ref{fig:LS_SNR_Threshold_plots_small}) follow the same pattern as the ER graphs and are presented in Appendix \ref{plot_appendix}.

\section{Discussion}
In this paper we studied the impact of sample size on LS reconstruction of noisy $k$-bandlimited graph signals. We showed theoretically and experimentally that reconstruction error is not necessarily monotonic in sample size - that at sufficiently low SNRs, reconstruction error can sometimes be improved by \emph{removing} a vertex from a sample set, even if the sample set was picked by a greedy optimal sampling scheme given a fixed sample size.

Our finding reveals that certain existing results in the literature for noiseless settings may not necessarily generalise to the noisy case. In addition, it further demonstrates the need to consider both optimal sample size selection and reconstruction methods at the same time. For example, the limitation of ordinary LS reconstruction may be mitigated by regularisation schemes such as that proposed in \cite{chamon2017greedy}.


Future work includes extending the analysis in this paper to cover other reconstruction operators such as GLR reconstruction, providing bounds on $\xi_1$ and $\xi_2$ to create noise-aware sample size bounds,  experimenting with other graph models such as Ring graphs or studying early-stopping schemes for LS reconstruction.


\bibliographystyle{IEEEtran}
\bibliography{main}
\clearpage

\appendices

\section{Under LS reconstruction, $\Delta_{1} \leq 0$}
\label{app:delta_1_non_positive}
For LS we have:  \[\matr{R}_{\set{S}} = \matr{U}_{k}(\matr{M}_{\set{S}}\matr{U}_{k})^{\dagger}.\]

\begin{lemma}
\label{lemma:frob_orthog}
    For any matrix $\matr{A}$, $||\matr{U}_{k}\matr{A}||^{2}_{F} = ||\matr{A}||^{2}_{F}$
\end{lemma}

\begin{proof}
    \begin{align*}
    ||\matr{U}_{k}\matr{A}||^{2}_{F} 
    &= \text{tr}(\matr{U}_{k}\matr{A}\matr{A}^{T}\matr{U}_{k}^{T})
    = \text{tr}(\matr{U}_{k}^{T}\matr{U}_{k}\matr{A}\matr{A}^{T}) \\
    &= \text{tr}(\matr{A}\matr{A}^{T}) = ||\matr{A}||^{2}_{F}.
\end{align*}
\end{proof}

\begin{lemma}
    \label{lemma:LS_xi_1_is_rank}
    For LS, $\xi_{1}(\set{S}) = k - \text{rank}(\matr{M}_{\set{S}}\matr{U}_{k})$.
\end{lemma}
\begin{proof}
Using Lemma \ref{lemma:frob_orthog},
    \begin{align*}
        \xi_{1}(\set{S}) &= || \matr{U}_{k} - \matr{R}_{\set{S}}\matr{M}_{\set{S}}\matr{U}_{k} ||^{2}_{F} \\
        &= || \matr{U}_{k} - \matr{U}_{k}(\matr{M}_{\set{S}}\matr{U}_{k})^{\dagger}\matr{M}_{\set{S}}\matr{U}_{k} ||^{2}_{F} \\
        &= || \matr{I}_{k} -(\matr{M}_{\set{S}}\matr{U}_{k})^{\dagger}\matr{M}_{\set{S}}\matr{U}_{k}||^{2}_{F}
    \end{align*}
    Let $\matr{\Pi} = (\matr{M}_{\set{S}}\matr{U}_{k})^{\dagger}\matr{M}_{\set{S}}\matr{U}_{k}$. $\matr{\Pi}$ is of the form $\matr{A}^{\dagger}\matr{A}$, so is a symmetric orthogonal projection onto the range of $(\matr{M}_{\set{S}}\matr{U}_{k})^{T}$ \cite[p.~258]{golub13}. Orthogonal projections are idempotent ($\matr{\Pi} = \matr{\Pi}^{2}$) hence have eigenvalues which are $0$ or $1$, and therefore $\text{tr}(\matr{\Pi}) = \text{rank}((\matr{M}_{\set{S}}\matr{U}_{k})^{T}) = \text{rank}(\matr{M}_{\set{S}}\matr{U}_{k})$. We then have:
    \begin{align*}
        \xi_{1}(\set{S}) &= ||\matr{I}_{k} - \matr{\Pi}||^{2}_{F} \\
        &= \text{tr}((\matr{I}_{k} - \matr{\Pi})(\matr{I}_{k} - \matr{\Pi})^{T})\\
        & =\text{tr}((\matr{I}_{k} - \matr{\Pi})(\matr{I}_{k} - \matr{\Pi})) \\
        &= \text{tr}(\matr{I}_{k} - 2\matr{\Pi} + \matr{\Pi}^{2}) \\
        &= \text{tr}(\matr{I}_{k} - \matr{\Pi}) \\
        &= \text{tr}(\matr{I}_{k}) - \text{tr}(\matr{\Pi}) \\
        &= k - \text{rank}(\matr{M}_{\set{S}}\matr{U}_{k}).
    \end{align*}
\end{proof}

\begin{lemma}
\label{lemma:delta_1_ls_is_discrete}
    For LS, $\Delta_{1}(\set{S},v) \in \{0, -1\}$.
\end{lemma}
\begin{proof}
    
Removing a vertex from $\set{S}$ removes a row from $\matr{M}_{\set{S}}\matr{U}_{k}$, reducing the rank by $0$ or $1$. 
\begin{align}
\Delta_{1}(\set{S},v) &= \xi_{1}(\set{S}) - \xi_{1}(\set{S} \backslash \{v\}) \nonumber \\
&= - \text{rank}(\matr{M}_{\set{S}}\matr{U}_{k}) + \text{rank}(\matr{M}_{\set{S} \backslash \{v\}}\matr{U}_{k}) \nonumber \\
&\in \{0, -1\}. \nonumber
\end{align}
\end{proof}

Non-positivity of $\Delta_{1}$ immediately follows from Lemma \ref{lemma:delta_1_ls_is_discrete}.

\section{Under LS reconstruction, $\Delta_{1} < 0 \iff \Delta_{2} > 0$}
\label{app:Delta_LS_opposite_signs}
We first need the following lemmas.

\begin{lemma}
\label{lemma:simplify_xi_2}
\begin{equation}
    \xi_{2}(\set{S}) = \sum_{\lambda_{i}^{\set{S}} \neq 0} { \frac{1}{\lambda_{i}^{\set{S}}}  }
\end{equation}
    where $\lambda_{i}^{\set{S}}$ is the $i^{th}$ eigenvalue of $(\matr{M}_{\set{S}}\matr{U}_{k})(\matr{M}_{\set{S}}\matr{U}_{k})^{T}$.
\end{lemma}
\begin{proof} By definition and Appendix \ref{app:delta_1_non_positive}, Lemma \ref{lemma:frob_orthog}
\begin{flalign*}
\xi_{2}(\set{S}) &= ||\matr{R}_{\set{S}}||^{2}_{F} \\
&= ||\matr{U}_{k}(\matr{M}_{\set{S}}\matr{U}_{k})^{\dagger}||^{2}_{F} \\
&= ||(\matr{M}_{\set{S}}\matr{U}_{k})^{\dagger}||^{2}_{F}
\end{flalign*}
    which is the sum of the squares of the singular values of $(\matr{M}_{\set{S}}\matr{U}_{k})^{\dagger}$ \cite[Corollary 2.4.3]{golub13}. The pseudoinverse maps the singular values of $\matr{M}_{\set{S}}\matr{U}_{k}$ onto the singular values of $(\matr{M}_{\set{S}}\matr{U}_{k})^{\dagger}$ in the following way \cite[Section 5.5.2]{golub13}:
    \begin{equation}
    \sigma_{i}((\matr{M}_{\set{S}}\matr{U}_{k})^{\dagger}) =
        \begin{cases}
            0 &\textrm{if }\sigma_{i}(\matr{M}_{\set{S}}\matr{U}_{k}) = 0 \\
            \sigma_{i}(\matr{M}_{\set{S}}\matr{U}_{k})^{-1} &\textrm{otherwise}
        \end{cases}
    \end{equation}
    and the squares of the singular values of $\matr{M}_{\set{S}}\matr{U}_{k}$ are $\lambda_{i}$ \cite[Eq. (8.6.1)]{golub13}. Summing them gives the result.
\end{proof}

\begin{lemma}
\label{lemma:square_to_rect_rank}
\begin{equation}
\text{rank}((\matr{M}_{\set{S}}\matr{U}_{k})(\matr{M}_{\set{S}}\matr{U}_{k})^T) = \text{rank}(\matr{M}_{\set{S}}\matr{U}_{k}) \leq k. \nonumber
\end{equation}
\end{lemma}
\begin{proof}

For the equality: $\textrm{rank}(\matr{M_{\set{S}}}\matr{U}_{k})$ is the number of strictly positive singular values it has \cite[Corollary 2.4.6]{golub13}. By \cite[Eq. (8.6.2)]{golub13}, this is the same as the number of strictly positive eigenvalues of $(\matr{M}_{\set{S}}\matr{U}_{k})(\matr{M}_{\set{S}}\matr{U}_{k})^T)$, which is $\text{rank}((\matr{M}_{\set{S}}\matr{U}_{k})(\matr{M}_{\set{S}}\matr{U}_{k})^T)$.

For the inequality: $\matr{M}_{\set{S}}\matr{U}_{k}$ has $k$ columns and so must have column rank less than or equal to $k$. Row rank being equal to column rank gives the result.
\end{proof}

\begin{lemma}
\label{lemma:LS_delta_1_is_0_improves_stability}
    For LS, $\Delta_{1} = 0 \iff \Delta_{2} \leq 0$.
\end{lemma}
\begin{proof}
    Note that $(\matr{M}_{\set{S} \backslash \{v\}}\matr{U}_{k})(\matr{M}_{\set{S} \backslash \{v\}}\matr{U}_{k})^{T}$ is a principal submatrix of $(\matr{M}_{\set{S}}\matr{U}_{k})(\matr{M}_{\set{S}}\matr{U}_{k})^{T}$. Write the eigenvalues of $(\matr{M}_{\set{S} \backslash \{v\}}\matr{U}_{k})(\matr{M}_{\set{S} \backslash \{v\}}\matr{U}_{k})^{T}$ as $\lambda_{1}, \dots, \lambda_{n}$ and the eigenvalues of $(\matr{M}_{\set{S}}\matr{U}_{k})(\matr{M}_{\set{S}}\matr{U}_{k})^{T}$ as $\mu_{1}, \dots \mu_{n+1}$. Then by Cauchy's Interlacing Theorem \cite[p.~59]{bhatia2013matrix},
     \begin{equation}
         0 \leq \mu_{1} \leq \lambda_{1} \leq \dots \leq \lambda_{n} \leq \mu_{n+1} \leq 1 \label{eq:MUUM_nonnegative}
     \end{equation} 
    where the outer bounds come from the fact that both matrices are principal submatrices of $\matr{U}_{k}\matr{U}_{k}^{T}$, an orthogonal projection matrix.
    
    \subsubsection{$\Delta_{1}= 0 \implies \Delta_{2} \leq 0$}
    $\Delta_{1} = 0$ implies the rank of $\matr{M}_{\set{S}}\matr{U}_{k}$ does not change with the removal of $v$, so neither does the rank of $(\matr{M}_{\set{S}}\matr{U}_{k})(\matr{M}_{\set{S}}\matr{U}_{k})^{T}$. As the rank is unchanged, $(\matr{M}_{\set{S}}\matr{U}_{k})(\matr{M}_{\set{S}}\matr{U}_{k})^{T}$ has one more zero-eigenvalue than $(\matr{M}_{\set{S} \backslash \{v\}}\matr{U}_{k})(\matr{M}_{\set{S} \backslash \{v\}}\matr{U}_{k})^{T}$. This means:
    \begin{align}
        \mu_{1} = 0 \label{eq:mu_1_is_zero}\\
        \lambda_{i} = 0 \iff \mu_{i+1} = 0 \label{eq:mu_lambda_biject_zeros}
    \end{align} 
    By Cauchy's Interlacing Theorem, $\lambda_{i} \leq \mu_{i+1}$ and so 
    \begin{equation}
        \frac{1}{\lambda_{i}} \geq \frac{1}{\mu_{i+1}} \text{ if } \lambda_{i} \neq 0 \text{ and } \mu_{i+1} \neq 0.  \label{eq:lambda_mu_ineq}
    \end{equation} Therefore
 \begin{equation}
     \sum_{\lambda_{i}^{\set{S}} \neq 0} { \frac{1}{\lambda_{i}^{\set{S}}}} \geq \sum_{\mu_{i}^{\set{S}} \neq 0} { \frac{1}{\mu_{i}^{\set{S}}}} \label{eq:eig_sum_ineq}
 \end{equation}
 as we have the same number of non-zero terms in each of these terms by (\ref{eq:mu_1_is_zero}) and (\ref{eq:mu_lambda_biject_zeros}), and the inequality is proved by summing over the non-zero terms using (\ref{eq:lambda_mu_ineq}).
Equation (\ref{eq:eig_sum_ineq}) is exactly 
\begin{equation}
   \xi_{2}(\set{S} \backslash \{v\}) \geq \xi_{2}(\set{S}).
\end{equation}
Rearranging gives $\Delta_{2} \leq 0$.

\subsubsection{$\Delta_{1} = 0 \impliedby \Delta_{2} \leq 0$} We prove the equivalent statement
\begin{equation}
    \Delta_{1} \neq 0 \implies \Delta_{2} > 0 .
\end{equation}
By Lemma \ref{lemma:delta_1_ls_is_discrete}, if $\Delta_{1} \neq 0$ then $ \Delta_{1} = -1$. This means that the rank of $\matr{M}_{\set{S}}\matr{U}_{k}$ is reduced by 1 by the removal of $v$, therefore $(\matr{M}_{\set{S}}\matr{U}_{k})(\matr{M}_{\set{S}}\matr{U}_{k})^{T}$ has one more non-zero eigenvalue than $(\matr{M}_{\set{S} \backslash \{v\}}\matr{U}_{k})(\matr{M}_{\set{S} \backslash \{v\}}\matr{U}_{k})^{T}$. This means:
\begin{align}
    \mu_{n+1} > 0 \label{eq:mu_nonzero} \\
    \lambda_{i} \neq 0 \iff \mu_{i} \neq 0 \label{eq:lambda_mu_match2}
\end{align}
By Cauchy's interlacing theorem, $\lambda_i \geq \mu_i$ and so
\begin{equation}
        \frac{1}{\lambda_{i}} \leq \frac{1}{\mu_{i}} \text{ if } \lambda_{i} \neq 0 \text{ and } \mu_{i} \neq 0.  \label{eq:lambda_mu_inv_ineq2}
    \end{equation}
    Let $I$ be the number of zero eigenvalues of $(\matr{M}_{\set{S}}\matr{U}_{k})(\matr{M}_{\set{S}}\matr{U}_{k})^T$. Then 
\begin{equation}
    \sum_{I\leq i \leq n} { \frac{1}{\lambda_{i}^{\set{S}}}} \leq \sum_{I\leq i \leq n} { \frac{1}{\mu_{i}^{\set{S}}}} <  \sum_{I \leq i \leq n+1} { \frac{1}{\mu_{i}^{\set{S}}}}. \label{eq:eig_sum_ineq2_intermediate}
\end{equation}
With the left inequality by matching terms via (\ref{eq:lambda_mu_match2}) and then summing over (\ref{eq:lambda_mu_inv_ineq2}), and the right inequality because (\ref{eq:mu_nonzero}) means $\frac{1}{\mu_{n+1}} > 0$. We then note the left and the right terms in this equality say:
\begin{equation}
     \sum_{\lambda_{i}^{\set{S}} \neq 0} { \frac{1}{\lambda_{i}^{\set{S}}}} < \sum_{\mu_{i}^{\set{S}} \neq 0} { \frac{1}{\mu_{i}^{\set{S}}}} \label{eq:eig_sum_ineq2}
 \end{equation}
or equivalently,
\begin{equation}
 \xi_{2}(\set{S} \backslash \{v\}) < \xi_{2}(\set{S}).   
\end{equation}
Rearranging gives $\Delta_{2} > 0$.
\end{proof}

We finally have the following:
\begin{lemma}    \label{lemma:LS_delta_1_improvement_means_delta_2_worse}
    For LS, $\Delta_{1} < 0 \iff \Delta_{2} > 0$.
\end{lemma}
\begin{proof}
    By 
    Lemma \ref{lemma:delta_1_ls_is_discrete} and Lemma \ref{lemma:LS_delta_1_is_0_improves_stability}.
\end{proof}

\section{Proof of Theorem \ref{main_ls}}
\label{proof_appendix_LS}
\begin{proof}
For brevity, we fix $\set{S}$ and $v$ and write $\Delta_{1} = \Delta_{1}(\set{S},v)$ and $\Delta_{2} = \Delta_{2}(\set{S},v) $.

Rearranging (\ref{eq:EMSE_decomp_into_delta}) gives us that $v$ improves $\set{S}$ if and only if 
\begin{equation}
    \Delta_{1} + \sigma^{2} \cdot \Delta_{2} > 0
\end{equation}
or equivalently if and only if 
\begin{equation}
    \Delta_{1} > - \sigma^{2} \cdot \Delta_{2} .
\end{equation}
By definition, $\sigma^{2} = \frac{k}{N \cdot \text{SNR}}$, so this condition is equivalent to
\begin{equation}
    \Delta_{1} > - \frac{k}{N \cdot \text{SNR}} \Delta_{2}
\end{equation}
and as SNR is strictly positive, this is equivalent to
\begin{equation}
    \text{SNR}\cdot \Delta_{1} > - \frac{k}{N} \Delta_{2}. \label{eq:thm_proof_linear_cond}
\end{equation}

We can now use the major lemmas from the previous appendices. By Lemma \ref{lemma:delta_1_ls_is_discrete}, we have two possible values of $\Delta_{1}(\set{S},v)$:
\subsection*{$\Delta_{1} = 0$:}
Lemma \ref{lemma:LS_delta_1_is_0_improves_stability} means $\Delta_{2} < 0$, so
\begin{equation}
    \Delta_{1} + \sigma^{2} \cdot \Delta_{2} = \sigma^{2} \cdot \Delta_{2} < 0
\end{equation}
and so $v$ does not improve $\set{S}$.

\subsection*{$\Delta_{1} = -1$:}
Eq. (\ref{eq:thm_proof_linear_cond}) simplifies to:
\begin{equation}
    - \text{SNR} > - \frac{k}{N} \Delta_{2}
\end{equation}
which is equivalent to
\begin{equation}
    \text{SNR} < \frac{k}{N} \Delta_{2}. \label{eq:delta_is_minus_1_cond_SNR}
\end{equation}

On the one hand, $v$ improves $\set{S}$ implies $\Delta_{1} = -1$, which implies (\ref{eq:delta_is_minus_1_cond_SNR}). On the other hand, (\ref{eq:delta_is_minus_1_cond_SNR}) implies $\Delta_{2} > 0$ which in turn implies $\Delta_{1}=-1$, which means (\ref{eq:delta_is_minus_1_cond_SNR}) implies (\ref{eq:thm_proof_linear_cond}), which implies $v$ improves $\set{S}$. 

Note that the right-hand side of (\ref{eq:delta_is_minus_1_cond_SNR}) is $\tau(\set{S},v)$; this completes the proof.
\end{proof}

\section{Proof of Theorem \ref{thm:main_existence_LS}}
We restate the theorem:
\label{app:LS_satisfied}
\begin{theorem}
    \label{thm:app:main_existence_LS}
    Consider any sequence of vertices $v_{1},\dots, v_{N}$ with no repeated vertices, and let $\set{S}_{i} = \{v_{1} , \dots , v_{i} \}$. Then there are exactly $k$ indices $I_{1}, \dots, I_{k}$ such that under LS reconstruction of a noisy $k$-bandlimited signal,
    \begin{equation}
        \forall 1\leq j \leq k: \tau(\set{S}_{I_{j}}, v_{I_{j}}) > 0
    \end{equation}
    and so for some $\text{SNR}>0$ removing $v_{I_{j}}$ would improve $\set{S}_{I_{j}}$.
\end{theorem}
\begin{proof}
\label{app:proof_main_existence_LS}
By Appendix \ref{proof_appendix_LS}, Lemma \ref{lemma:LS_xi_1_is_rank}: 
\begin{equation}
    \xi_{1}(\set{S}_{i}) = k -\text{rank}(\matr{M}_{\set{S}_{i}}\matr{U}_{k}).
\end{equation}
By Appendix \ref{proof_appendix_LS}, Lemma \ref{lemma:delta_1_ls_is_discrete}, $\Delta_{1} \in \{0, -1\}$ and as $\text{rank}(\matr{U}_{k}) = k$, $\xi_{1}(\set{S}_{N}) = 0$. As $\xi_{1}(\set{S}_{0}) = k$, we must have exactly $k$ indices for which $\Delta_{1}(\set{S}_{i}, v_{i}) = -1$, and by Appendix \ref{proof_appendix_LS}, Lemma \ref{lemma:LS_delta_1_is_0_improves_stability} we have exactly $k$ indices for which $\Delta_{2}(\set{S}_{i}, v_{i}) > 0$. As $\tau (\set{S}_{i}, v_{i})= \frac{k}{N} \Delta_{2}(\set{S}_{i}, v_{i})$, we're done.
\end{proof}

\section{Proof of Theorem \ref{thm:noiseless_optimality_means_noise_sensitivity}}
\label{app:optimal_noiseless_schemes_immediately_satisfy}

\begin{proof}

By Appendix \ref{proof_appendix_LS}, Lemma \ref{lemma:LS_xi_1_is_rank}, the noiseless error
\begin{equation}
    \xi_{1}(\set{S}) = k -\text{rank}(\matr{M}_{\set{S}}\matr{U}_{k})
\end{equation}
must be 0, as we can perfectly reconstruct any $k$-bandlimited signal. Therefore, $\text{rank}(\matr{M}_{\set{S}}\matr{U}_{k}) = k$.

$\matr{M}_{\set{S}}\matr{U}_{k}$ is a $k \times k$ matrix of full rank, so its rows must be linearly independent. 
Any subset of linearly independent rows is linearly independent, so for any non-empty $\set{R} \subset \set{S}$, $\matr{M}_{\set{R}}\matr{U}_{k}$ has linearly independent rows.

Greedy schemes pick increasing sample sets: that is, if asked to pick a vertex sample set $\set{S}_{m}$ of size $m$ for $m < k$ and a sample set $\set{S}$ of size $k$, $\set{S}_{m} \subset \set{S}$. Therefore for any sample set $\set{S}_{m}$ of size $m \leq k$ picked by the scheme, $\matr{M}_{\set{S}_{m}}\matr{U}_{k}$ has independent rows.

If $\matr{M}_{\set{S}_{m}}\matr{U}_{k}$ has independent rows, then removal of any row (corresponding to removing any vertex) reduces its rank by 1; that is,
\begin{equation}
    \forall m \leq k: \enskip \forall v \in \set{S}_{m}: \enskip \Delta_{1}(\set{S}_{m},v) = -1
\end{equation}
Then, by Appendix \ref{proof_appendix_LS}, Lemma \ref{lemma:LS_delta_1_improvement_means_delta_2_worse},
\begin{equation}
    \forall m \leq k: \enskip \forall v \in \set{S}_{m}: \enskip \Delta_{2}(\set{S}_{m},v) > 0 
\end{equation}
and as $\tau(\set{S}_{m},v) = \frac{k}{N}\Delta_{2}(\set{S}_{m},v)$ and $\frac{k}{N} > 0$,
\begin{equation}
     \forall m \leq k: \enskip \forall v \in \set{S}_{m}: \enskip \tau(\set{S}_{m},v) > 0. 
\end{equation}
This proves (\ref{eq:greedy_sampling_first_k}). 

As $\matr{M}_{\set{S}_{k}}\matr{U}_{k}$ has $k$ independent rows, it is of rank k. Adding further rows can't decrease its rank, so for $m' > k$, $\textrm{rank}(\matr{M}_{\set{S}_{m'}}\matr{U}_{k}) \geq k$. As $\matr{U}_{k}$ is of rank $k$, $\textrm{rank}(\matr{M}_{\set{S}_{m'}}\matr{U}_{k}) \leq k$. This means for all samples sizes $m' > k$, $\textrm{rank}(\matr{M}_{\set{S}_{m'}}\matr{U}_{k}) = k$. This says that further additions of rows do not change rank; that is:
\begin{equation} 
   \forall m' > k: \enskip \forall v \in \set{S}_{m'} \backslash \set{S}_{k}: \enskip \Delta_{1}(\set{S}_{m'},v) = 0
\end{equation}
Then, by Appendix \ref{proof_appendix_LS}, Lemma \ref{lemma:LS_delta_1_is_0_improves_stability},
\begin{equation}
    \forall m' > k: \enskip \forall v \in \set{S}_{m'} \backslash \set{S}_{k}: \enskip \Delta_{2}(\set{S}_{m'},v) \leq 0 
\end{equation}
and, like for (\ref{eq:greedy_sampling_first_k}, as $\tau(\set{S}_{m},v) = \frac{k}{N}\Delta_{2}(\set{S}_{m},v)$ and $\frac{k}{N} > 0$,
\begin{equation}
     \forall m' > k: \enskip \forall v \in \set{S}_{m'}\backslash \set{S}_{k} : \enskip \tau(\set{S}_{m'},v) \leq 0. 
\end{equation}
This proves (\ref{eq:greedy_sampling_over_k}). 
\end{proof}

\section{Proof of Remark \ref{remark:ADE_are_noiseless_optimal}}
\label{app:proof_of_remark_ADE_are_noiseless_optimal}

\subsection*{A-Optimality}
A-optimality depends on the existence of the inverse of $(\matr{M}_{\set{S}}\matr{U}_{k})(\matr{M}_{\set{S}}\matr{U}_{k})^{T}$ existing, which requires it to be of full rank. By Appendix \ref{proof_appendix_LS}, Lemma \ref{lemma:square_to_rect_rank}, if an A-optimal scheme picks a set $\set{S}$ of size $k$, then $\text{rank}(\matr{M}_{\set{S}}\matr{U}_{k}) = k$. Therefore, $\set{S}$ is a uniqueness set \cite{{anis2016efficient}} and can perfectly reconstruct any $k$-bandlimited signal.

\subsection*{D- and E-optimality}
We show that for sample sizes less than $k$ we can always pick a row which keeps $(\matr{M}_{\set{S}}\matr{U}_{k})(\matr{M}_{\set{S}}\matr{U}_{k})^{T}$ full rank (of rank $|\set{S}|$), and that D- and E-optimal schemes do so.

By Appendix \ref{proof_appendix_LS}, Lemma \ref{lemma:square_to_rect_rank}, $\text{rank}(\matr{M}_{\set{S}}\matr{U}_{k})(\matr{M}_{\set{S}}\matr{U}_{k})^{T} = \text{rank}(\matr{M}_{\set{S}}\matr{U}_{k})$, so we only need to ensure $\text{rank}(\matr{M}_{\set{S}}\matr{U}_{k}) = |\set{S}|$.

We proceed by induction: given $\set{S}_{1}$ with $|\set{S}_{1}| = 1$, $\text{rank}(\matr{M}_{\set{S}_{1}}\matr{U}_{k}) = 1$. Assume that for $\set{S}_{i}$ with $|\set{S}_{i}| = i < k$, $\text{rank}(\matr{M}_{\set{S}_{i}}\matr{U}_{k}) = i$. As $\text{rank}(\matr{U}_{k}) = k$ and $i < k$, we can find a row to add to $\matr{M}_{\set{S}_{i}}\matr{U}_{k}$ which will increase its rank (else all other rows would lie in the $i$-dimensional space spanned by the rows of $\matr{M}_{\set{S}_{i}}\matr{U}_{k}$, which would imply $\text{rank}(\matr{U}_{k}) = i$, which is a contradiction as $i < k$). Adding the vertex which corresponds to the row to $\set{S}_{i}$ gives $\set{S}_{i+1}$ with $\text{rank}(\matr{M}_{\set{S}_{i+1}}\matr{U}_{k}) = i+1$.

We have shown that we can greedily choose to keep $\text{rank}(\matr{M}_{\set{S}}\matr{U}_{k}) = |\set{S}|$. We now show that D- and E-optimal schemes do so. The eigenvalues of $(\matr{M}_{\set{S}}\matr{U}_{k})(\matr{M}_{\set{S}}\matr{U}_{k})^{T}$ are non-negative (see Appendix \ref{proof_appendix_LS}, Eq. (\ref{eq:MUUM_nonnegative})), so any invertible $(\matr{M}_{\set{S}}\matr{U}_{k})(\matr{M}_{\set{S}}\matr{U}_{k})^{T}$ will have a strictly positive determinant and minimum eigenvalue, which are preferable under the D- and E- optimality criterion respectively to a non-invertible $(\matr{M}_{\set{S}}\matr{U}_{k})(\matr{M}_{\set{S}}\matr{U}_{k})^{T}$, which has a determinant and minimum eigenvalue of 0. Therefore, greedy D- and E- optimal sampling schemes will make sure $(\matr{M}_{\set{S}}\matr{U}_{k})(\matr{M}_{\set{S}}\matr{U}_{k})^{T}$ is invertible, and thus keep $\text{rank}(\matr{M}_{\set{S}}\matr{U}_{k}) = |\set{S}|$ for $|\set{S}| \leq k$. Therefore when D- and E- optimal schemes pick a set $\set{S}$ of size $k$, $\text{rank}(\matr{M}_{\set{S}}\matr{U}_{k}) = k$. Therefore, $\set{S}$ is a uniqueness set \cite{{anis2016efficient}} and can perfectly reconstruct any $k$-bandlimited signal.

\section{Additional Results}
\label{plot_appendix}
We show thresholds for the ER, BA and SBM graphs with 100 vertices (Fig. \ref{fig:LS_SNR_Threshold_plots_small}). We also present MSE plots for the larger BA (Fig \ref{LS_BA_MSE_fig}) and SBM (Fig \ref{LS_SBM_MSE_fig}) graphs.

\begin{figure*}%
    \centering
    \begin{subfigure}{0.6\columnwidth}
    \includegraphics[width=\columnwidth]{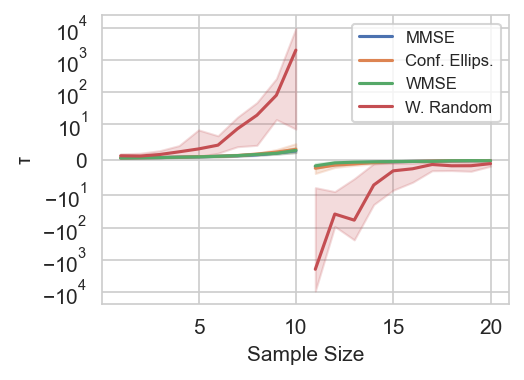}
    \caption{Erdos-Renyi}
    \label{snr_ER_small}
    \end{subfigure}\hfill
    \begin{subfigure}{0.6\columnwidth}
    \includegraphics[width=\columnwidth]{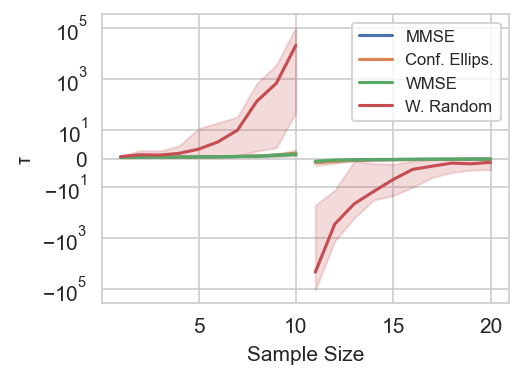}%
    \caption{Barabasi-Albert}%
    \label{snr_BA_small}%
    \end{subfigure}\hfill%
    \begin{subfigure}{0.6\columnwidth}
    \includegraphics[width=\columnwidth]{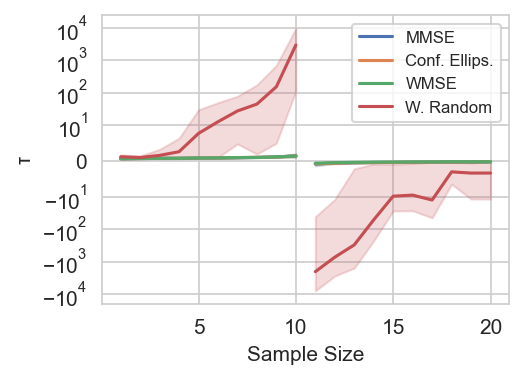}%
    \caption{SBM}%
    \label{snr_SBM_small}%
    \end{subfigure}%
    \caption{$\tau$ for different random graph models under LS reconstruction (\#vertices = 100, bandwidth = 10)}
\label{fig:LS_SNR_Threshold_plots_small}
\end{figure*}

\begin{figure*}%
    \label{LS_BA_MSE_fig}
    \centering
    \begin{subfigure}{0.6\columnwidth}
    \includegraphics[width=\columnwidth]{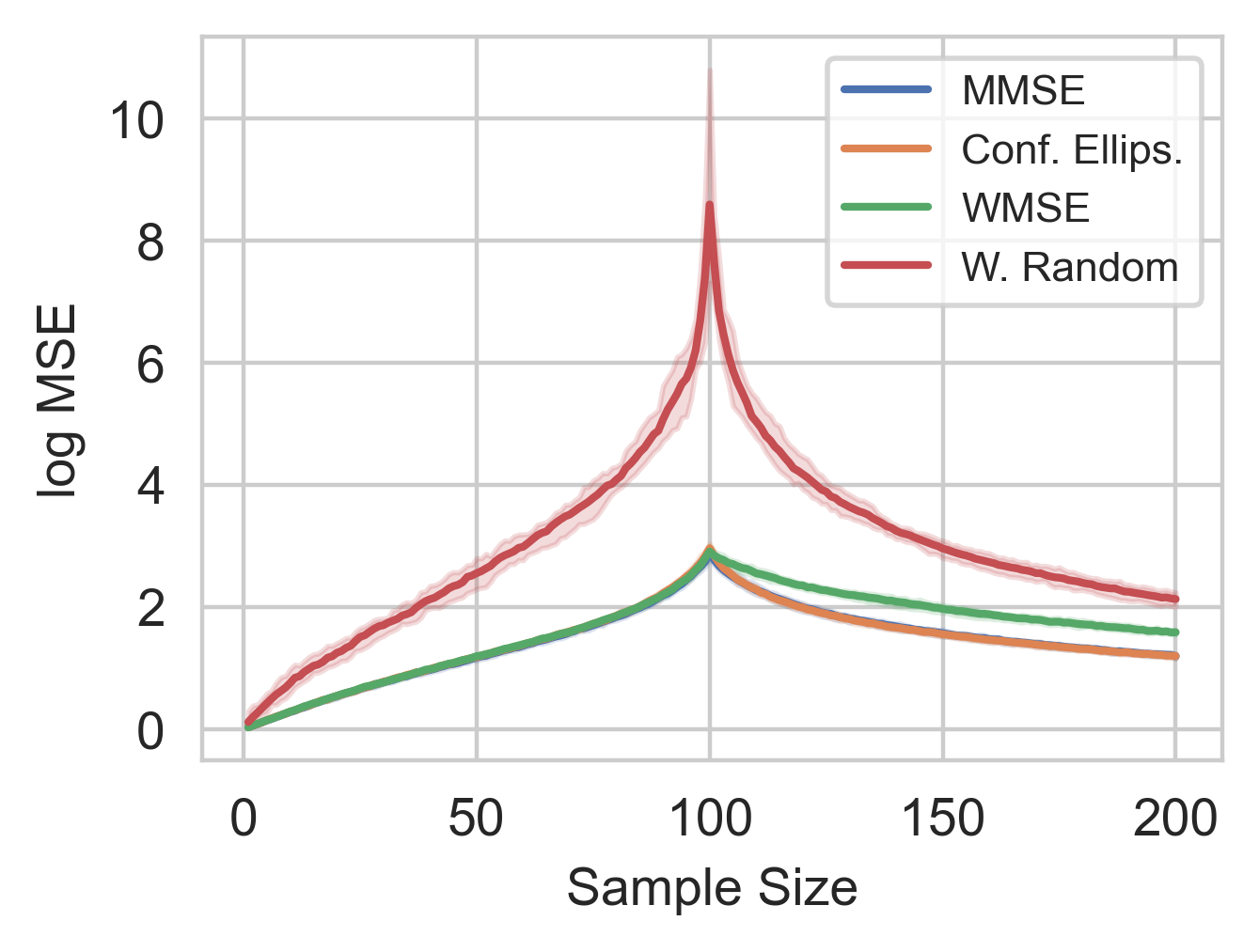}
    \caption{SNR = $10^{-1}$}
    \label{BA_MSE_subfiga}
    \end{subfigure}\hfill
    \begin{subfigure}{0.6\columnwidth}
    \includegraphics[width=\columnwidth]{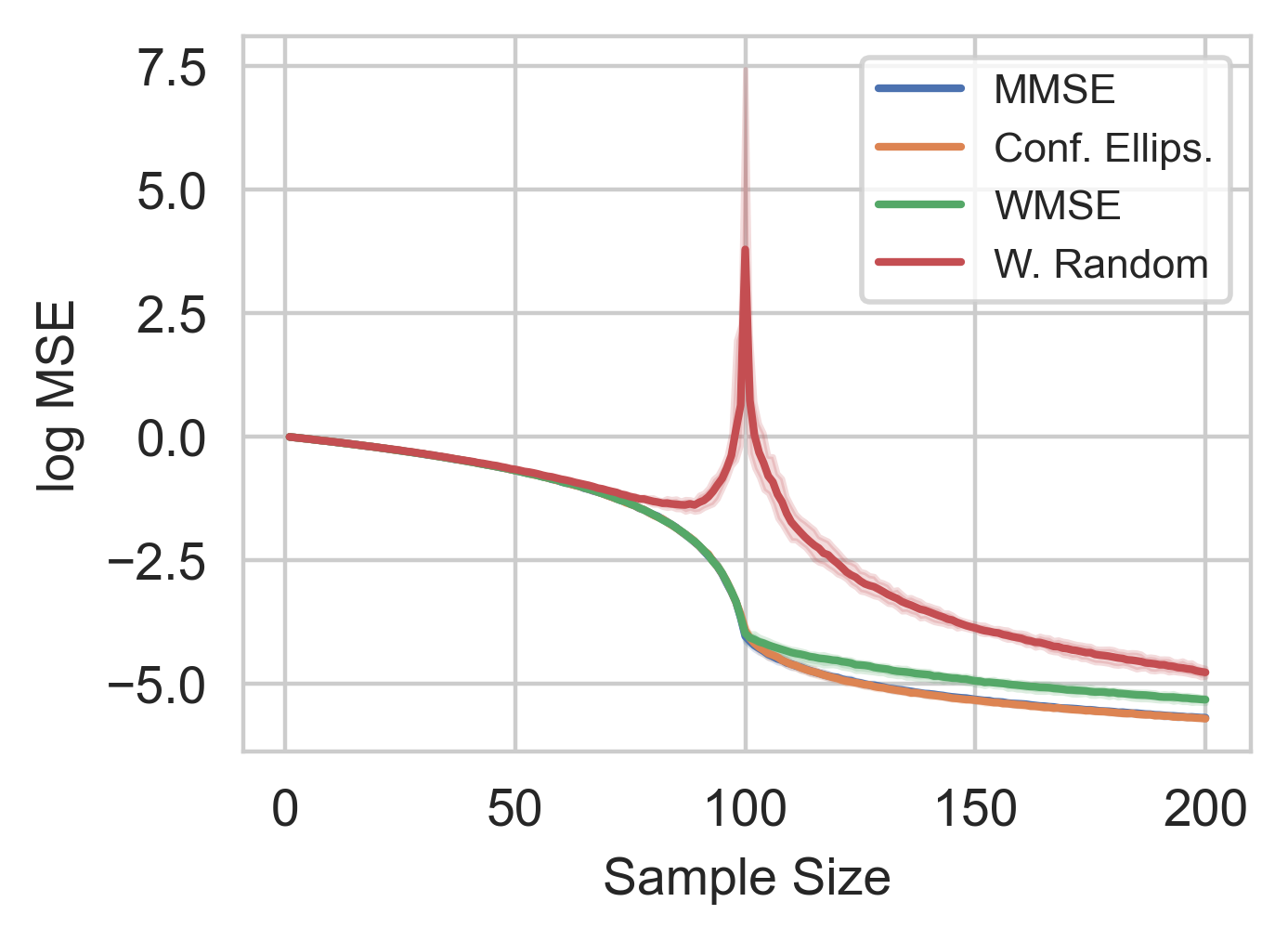}%
    \caption{SNR = $10^{2}$}%
    \label{BA_MSE_subfigb}%
    \end{subfigure}\hfill%
    \begin{subfigure}{0.6\columnwidth}
    \includegraphics[width=\columnwidth]{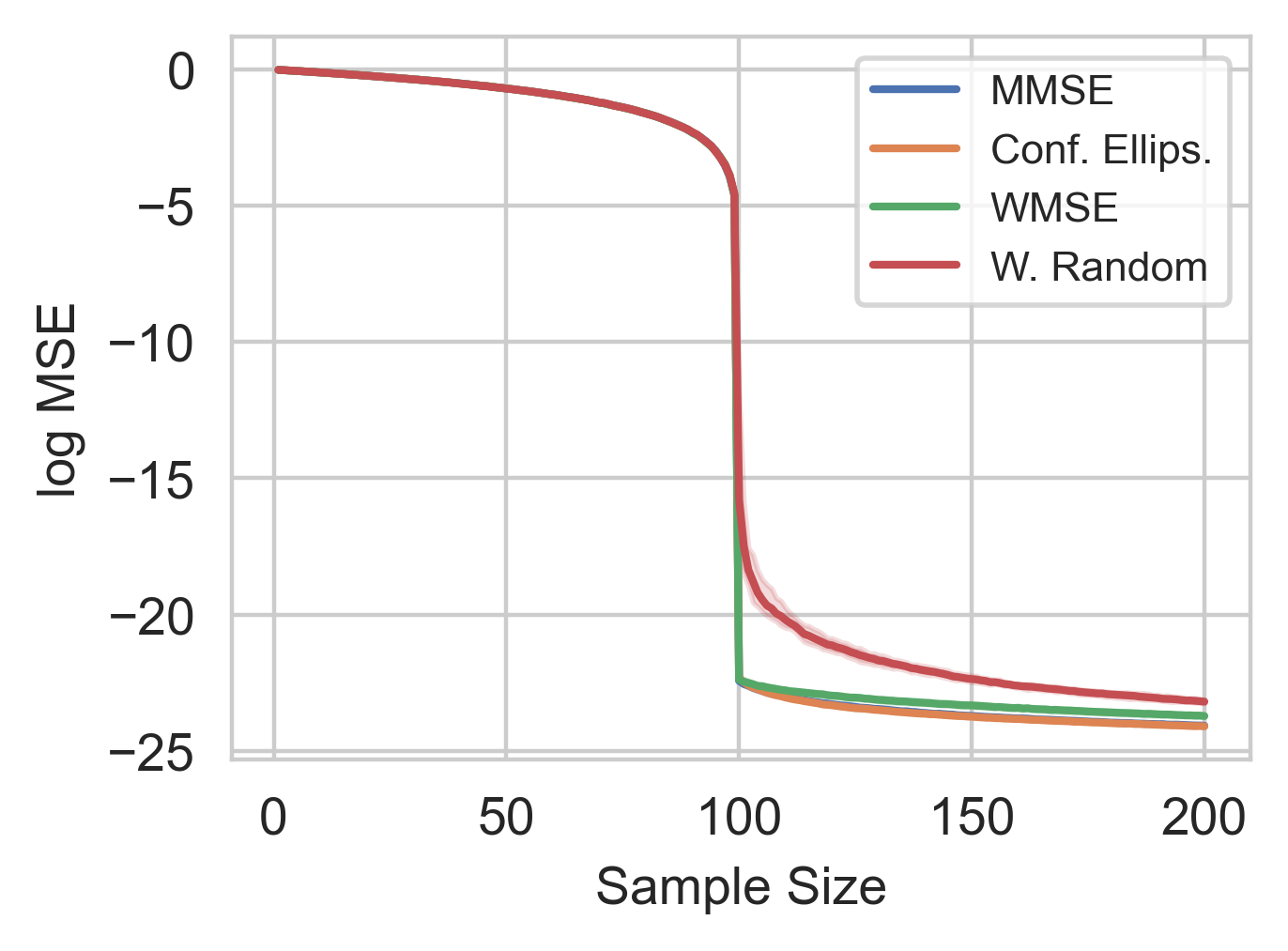}%
    \caption{SNR = $10^{10}$}%
    \label{BA_MSE_subfigc}%
    \end{subfigure}%
    \caption{Average MSE for LS reconstruction on BA Graphs (\#vertices=1000, bandwidth = 100) with different SNRs}
\label{LS_BA_MSE_fig}
\end{figure*}

\begin{figure*}%
    \label{LS_SBM_MSE_fig}
    \centering
    \begin{subfigure}{0.6\columnwidth}
    \includegraphics[width=\columnwidth]{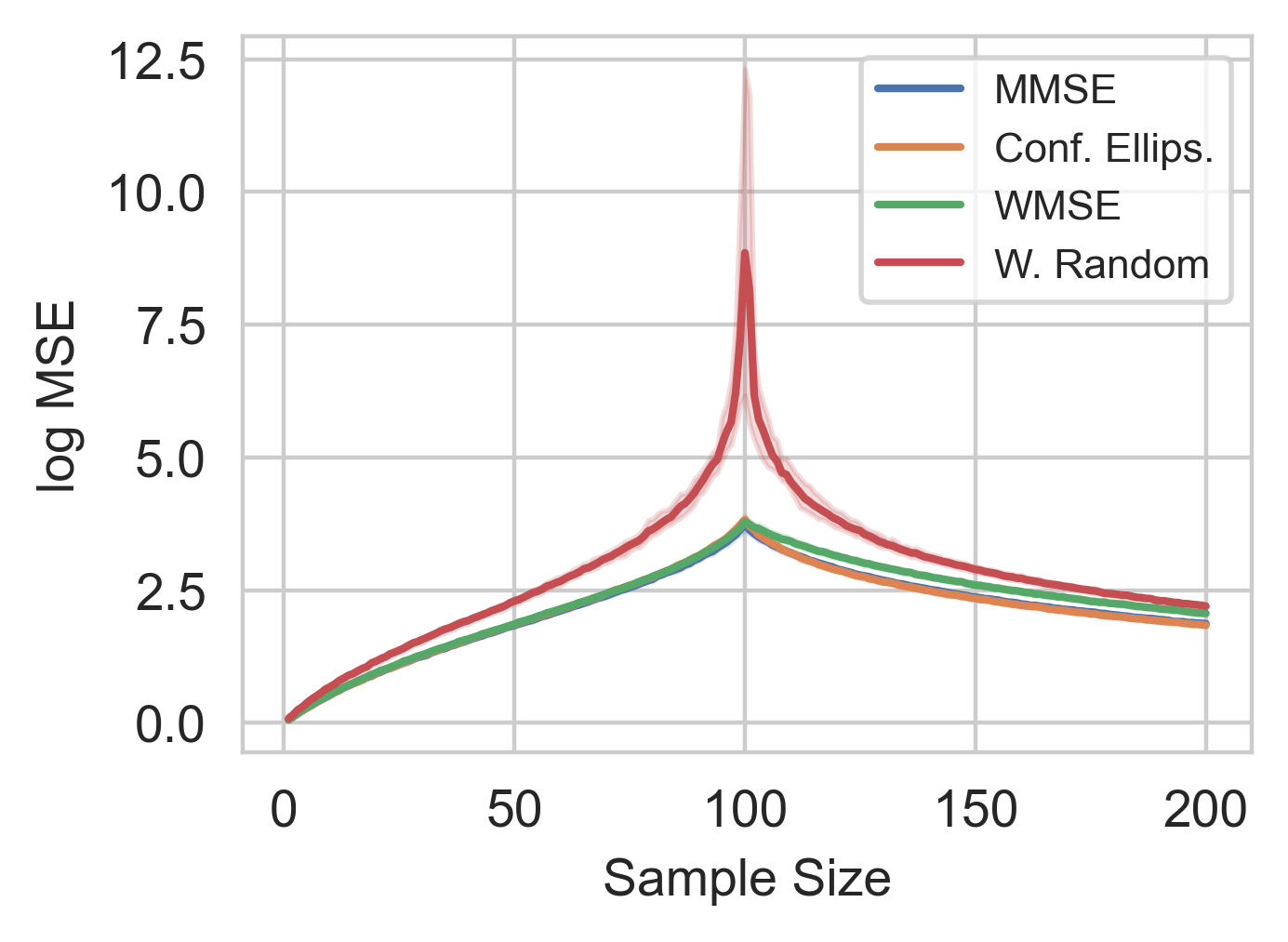}
    \caption{SNR = $10^{-1}$}
    \label{SBM_MSE_subfiga}
    \end{subfigure}\hfill
    \begin{subfigure}{0.6\columnwidth}
    \includegraphics[width=\columnwidth]{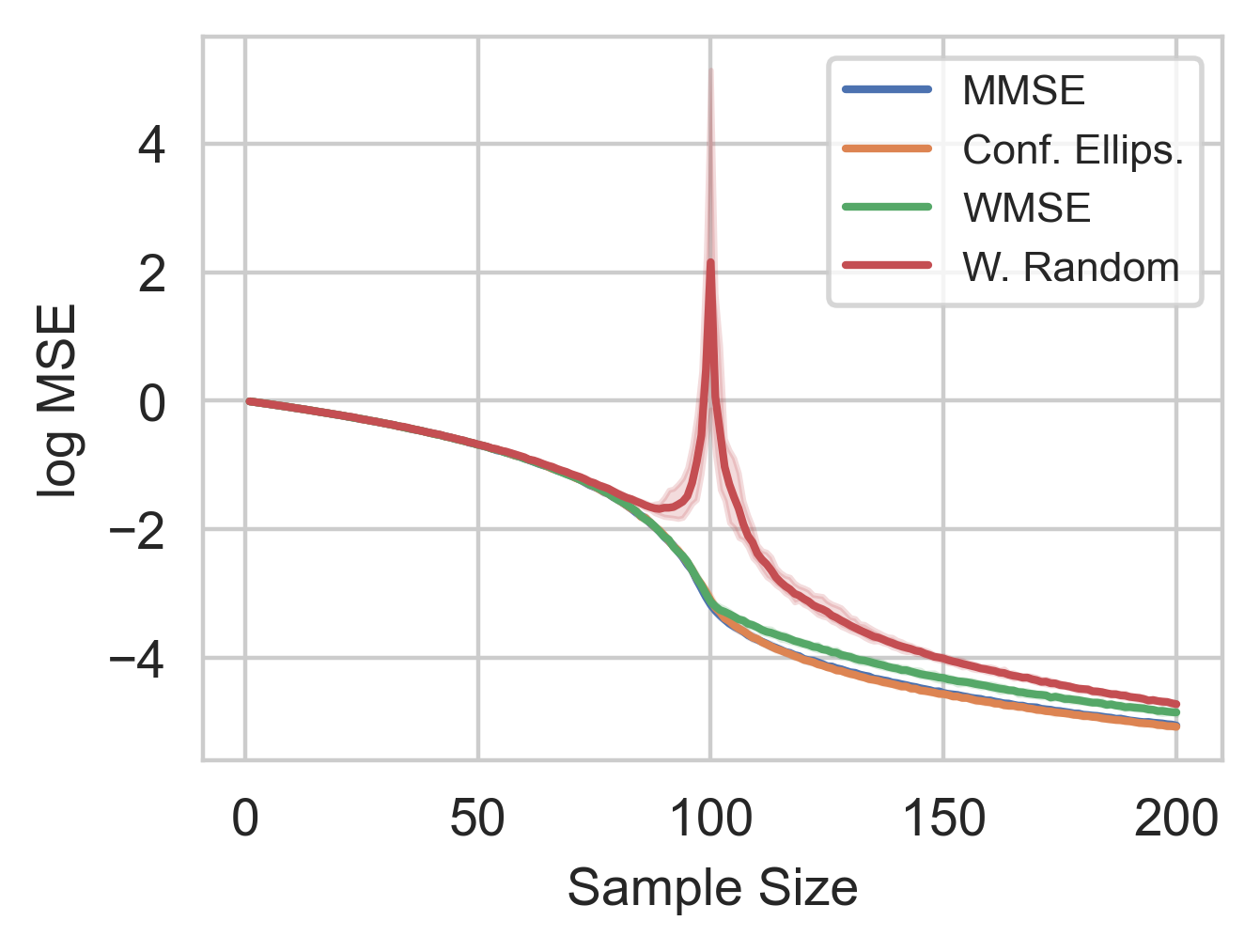}%
    \caption{SNR = $10^{2}$}%
    \label{SBM_MSE_subfigb}%
    \end{subfigure}\hfill%
    \begin{subfigure}{0.6\columnwidth}
    \includegraphics[width=\columnwidth]{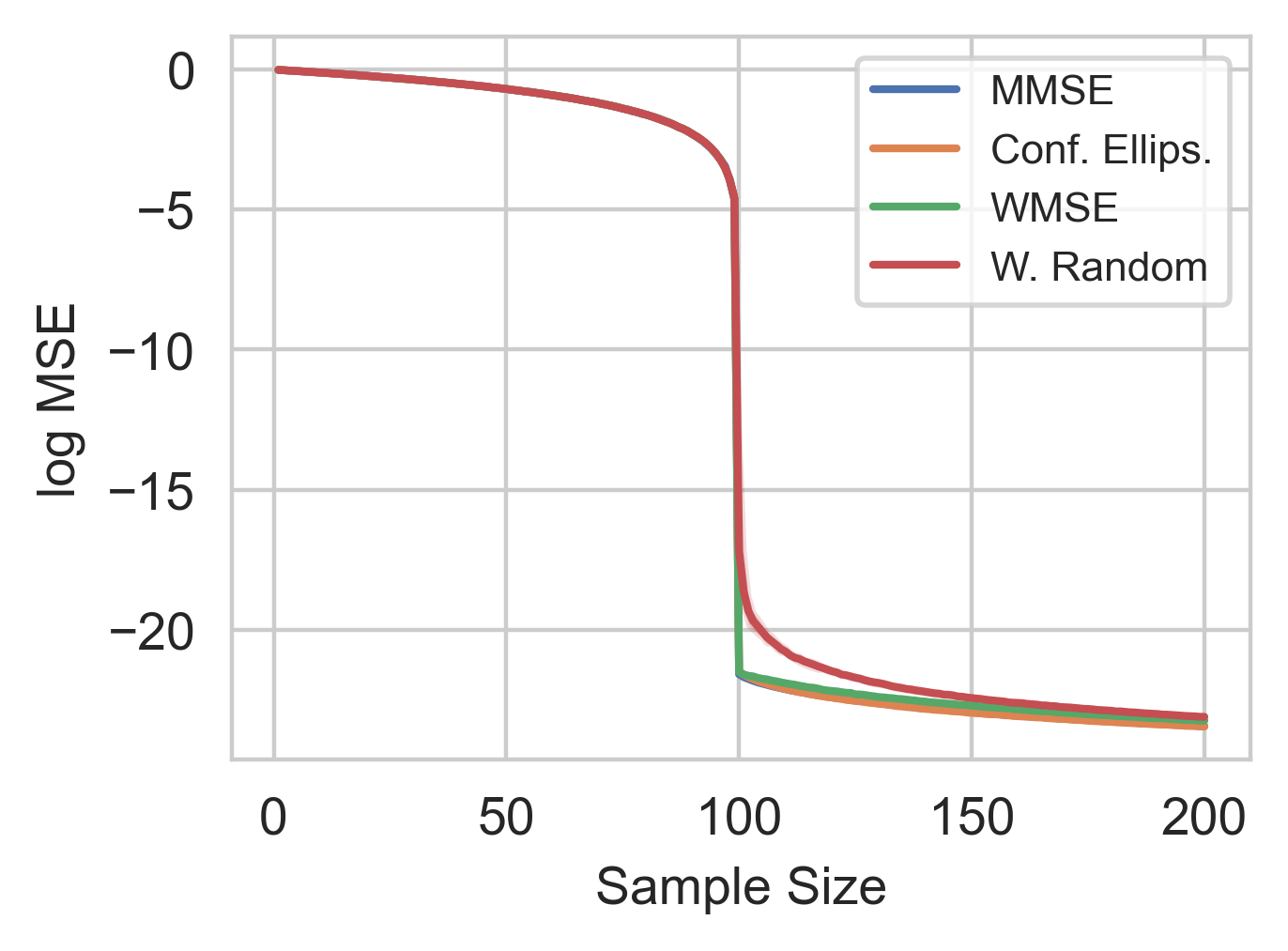}%
    \caption{SNR = $10^{10}$}%
    \label{SBM_MSE_subfigc}%
    \end{subfigure}%
    \caption{Average MSE for LS reconstruction on SBM Graphs (\#vertices=1000, bandwidth = 100) with different SNRs}
\label{LS_SBM_MSE_fig}
\end{figure*}


\end{document}